\documentclass{llncs}

\usepackage{amssymb}
\usepackage{amsmath}
\usepackage{textcomp}
\usepackage{array}
\usepackage{lineno}
\usepackage{color}
\usepackage{comment}

\usepackage{graphicx}

\newcommand{\bra}[1]{\langle #1|}
\newcommand{\ket}[1]{|#1\rangle}
\newcommand{\braket}[2]{\langle #1|#2\rangle}

\newcommand{\dollar}[0]{\$}

\newcommand{\mymatrix}[2]{\left( \begin{array}{#1} #2\end{array} \right)}

\newcommand{\myvector}[1]{\mymatrix{c}{#1}}

\newtheorem{fact}{Fact}

\title{Can one quantum bit separate any pair of words with zero-error?}

\author{Aleksandrs Belovs\inst{1}$^,$\thanks{Belovs was supported by FP7 FET Proactive project QALGO.} \and Juan Andres Montoya\inst{2} \and Abuzer Yakary{\i}lmaz\inst{3}$^,$\thanks{Yakary{\i}lmaz was partially supported by CAPES with grant 88881.030338/2013-01. Moreover, some parts of the work were done while Yakary{\i}lmaz was visiting Bogot\'{a}, Colombia in December 2014.}}

\institute{CWI, Amsterdam, the Netherlands
		\and 
		Departamento de Matem\'{a}ticas, Universidad Nacional de Colombia, \\ Bogot\'{a}, Colombia     
       \and
National Laboratory for Scientific Computing, \\ Petr\'{o}polis, RJ, 25651-075, Brazil
\email{stiboh@gmail.com,jamontoyaa@unal.edu.co,abuzer@lncc.br}
}

\authorrunning{A. Belovs \and J. A. Montoya \and A. Yakary{\i}lmaz} 

\begin{document}

\maketitle

\begin{abstract}
Determining the minimum number of states required by a finite automaton to separate a given pair of different words is an important problem. In this paper, we consider this problem for quantum automata (QFAs). 
We show that 2-state QFAs can separate any pair of words in nondeterministic acceptance mode and conjecture that they can separate any pair also with zero-error. 
Then, we focus on (a more general problem) separating a pair of two disjoint finite set of words. We show that QFAs can separate them efficiently in nondeterministic acceptance mode, i.e. the number of states is two to the power of the size of the small set. Additionally, we examine affine finite automata (AfAs) and show that two states are enough to separate any pair with zero-error. Moreover, AfAs can separate any pair of disjoint finite sets of words with one-sided bounded error efficiently like QFAs in nondeterministic mode.
\\
\textbf{Keywords:} \textit{quantum finite automaton, affine finite automaton, zero--error, bounded--error, nondeterminism, promise problems, succinctness}
\end{abstract}

\section{Introduction}

Separating two words is the simplest non-trivial promise problem: given two different words $x$ and $y$, accept $x$ and reject $y$.
Nonetheless, it is still not known what is the size of a deterministic finite automaton (DFA) solving this problem for the worst pair of words $x$ and $y$ of length $n$.
This problem was formulated by Goral\v c\'ik and Koubek \cite{GK86}.
The best known upper bound $O(n^{2/5}\log^{3/5} n)$ is due to Robson \cite{Rob89,Rob96}.
The best known lower bound is only $\Omega(\log n)$.
The same bound also holds for \textit{non-deterministic} finite automata~\cite{DESW11}. 
Recently, some alternative models of automata were investigated in \cite{YakM15A}.

In this paper, we focus on quantum finite automata (QFAs).
Recently it was shown that the state efficiency of zero-error QFAs over bounded-error probabilistic finite automata (or some other classical automata models) cannot be bounded if we focus on promise problems \cite{AY12,BMP14,RasY14A,GefY15A,GQZ15}. Since separating a pair of different words (or disjoint finite languages) is a very special promise problem, we find it interesting to ask how efficient can QFAs be on these problems.
We show that 2-state QFAs can separate any pair of words if we allow nondeterministic acceptance mode.
We also conjecture, and give some evidence that 2-state QFAs can separate any pair \textit{also 
with zero error probability}!  

Then, we focus on (a more general problem) separating a pair of two disjoint finite set of words. We show that QFAs can separate them efficiently in nondeterministic acceptance mode with the number of states exponential in the size of the smallest set.

Very recently, affine finite automaton (AfA) was introduced as a quantum-like non-linear generalization of PFA that can use negative transition values \cite{DCY16A}, and, they were shown to be more powerful than both QFAs and PFAs with bounded and unbounded error. Moreover, they can be more state efficient than QFAs and PFAs \cite{VilY16A}. Therefore, we also investigate AfAs in our context and we indeed obtain better results than QFAs. 2-state AfAs can separate any pair with zero-error. Moreover, they can separate any given word from the rest of words with one-sided bounded error. In the case of separating a pair of disjoint finite languages, AfAs can efficiently separate them with zero-error if one language is a singleton and with one-sided bounded-error, otherwise, where the number of states is two the power of the big (small) set in the first (second) case.

In the next section, we provide the necessary background. The results on separating pairs is given in Section \ref{sec:pairs}. It also includes our bounded-error AfA algorithms recognizing singleton languages. Then, the results on separating two finite sets are presented in Section \ref{sec:finite-sets}.

\section{Background}
We refer the reader to \cite{SayY14A} for a pedagogical introduction to quantum finite automata (QFAs), to \cite{AY15A} for an inclusive survey on QFAs, and to \cite{NC10} for a complete reference on quantum computation \cite{NC10}. For the basics of affine systems and affine finite automata (AfAs), we refer the reader to \cite{DCY16A}. 

We denote the alphabet by $ \Sigma $, not containing right end-marker $  \dollar $, throughout the paper. For any given word $ x \in \Sigma $, $ |x| $ represents the length of $x$, $ |x|_\sigma $ represents the number of occurrences of symbol $ \sigma $ in $ x $, and $ x_j $ represents the $ j $-th symbol of $ x $, where $ \sigma \in \Sigma $ and $ 1 \leq j \leq |x| $. As a special case, if $ |\Sigma| = 1 $, then the automaton and languages can be called unary.

\subsection{Easy and hard pairs}

Throughout the paper, a pair of words $ (x,y) $ refers two different words defined on the same alphabet. A pair of words $ (x,y) $ is called \textit{easy} if $ x $ and $ y $ has different numbers of occurrences of a symbol, i.e. $ \exists \sigma \in \Sigma \mbox{ such that } |x|_\sigma \neq |y|_\sigma $.  Otherwise, the pair is called \textit{hard}. Remark that any pair with different lengths (and so any unary pair) is easy.

Any hard pair defined on an alphabet with at least three elements can be mapped to a binary hard pair as follows. Let $ (x,y) $ be a hard pair defined on $ \{ \sigma_1,\ldots,\sigma_k \} $ for some $ k>2 $. Since the pair is hard, we have
\[
	|x|_{\sigma_i} = |y|_{\sigma_i}
\]
for each $ 1 \leq i \leq k $. Then there should be an index $ j $ ($ 1 \leq j \leq |x| = |y| $) such that $ x_j = \sigma_i \neq y_j = \sigma_{i'} $ for $ i \neq i' $. If we delete all the other symbols and keep only $ \sigma_i $s and $ \sigma_{i'} $s in $x$ and $ y$, we obtain two new words: $ x' $ and $ y' $, respectively. It is clear that $ (x',y') $ is a hard pair. So, instead of separating the hard pair $ (x,y) $, we can try to separate $ (x',y') $. Algorithmically, we apply the identity operators on the symbols other than $ \sigma_i $ and $ \sigma_{i'} $. Hence, unless otherwise specified, we focus on only unary and binary words throughout the paper.

\subsection{QFAs}

Quantum finite automaton (QFA) is a non-trivial generalization of probabilistic finite automaton \cite{Hir10,YS11A}. Here we give the definition of the known simplest QFA model, called Moore-Crutchfield QFAs (MCQFAs) \cite{MC00} since we can present our results (and our conjecture) based on this model.

An $n$-state MCQFA $ M $, which operates on $ n $-dimensional Hilbert space ($ \mathcal{H}_n $, i.e. $ \mathbb{C}^n $ with the inner product)   is a 5-tuple
\[
	M=(Q,\Sigma,\{ U_{\sigma} \mid \sigma \in \Sigma \},\ket{u_0},Q_a),
\]
where $ Q = \{ q_1,\ldots,q_n \} $ is the set of states, $ U_{\sigma} \in \mathbb{C}^{n \times n} $ is a unitary transition matrix whose $ (i,j) $th entry represent the transition amplitude from the state $ q_j $ to the state $ q_i $ when reading symbol $ \sigma \in \Sigma $ ($ 1 \leq i,j \leq n $), $ \ket{u_0} \in \mathbb{C}^n $ is the column vector representing the initial quantum state, and $ Q_a \subseteq Q $ is the set of accepting states. The basis of $ \mathcal{H}_n $ is formed by $ \{ \ket{q_j} \mid 1 \leq j \leq n \} $ where $ \ket{q_j} $ has 1 at the $ j $-th entry and 0s in the remaining entries. At the beginning of the computation, $ M $ is in $ \ket{u_0} $, either one of the basis states or a superposition (a linear combination) of basis states. Let $ x \in \Sigma^* $ be a given input word. During reading the input $x$ from left to right symbol by symbol, the quantum state of $ M $ is changed as follows:
\[
	\ket{u_j} = U_{x_j} \ket{u_{j-1}},
\] 
where $ 1 \leq j \leq |x| $. After reading the whole word, the quantum state is measured to determine whether $M$ is in an accepting state or not (a measurement on computational basis). Let the final quantum state,  represented as $\ket{u_f^x}$ or $ \ket{u_f} $,  have the following amplitudes
\[
	\ket{u_f^x} = \ket{u_f} = \ket{u_{|w|}} = \left( \begin{array}{c}
		\alpha_1\\ \alpha_2 \\ \vdots \\ \alpha_n 
	\end{array}	 \right).
\] Since the probability of observing $ j $th state is $ | \alpha_j |^2 $, the input is accepted with probability $ \sum_{q_j \in Q_a} | \alpha_j |^2  $. 

\subsection{AfAs}

An affine finite automaton \cite{DCY16A} can be in an affine state that can be represented as a column vector over real numbers where the summation of all entries are equal to 1. The evolution of an AfA is governed by affine transformations that preserve the summation of vectors, i.e. each column of an affine transformation is an affine state. To retrieve information from an AfA, a measurement-like operator called weighting operator is applied (the details are given below).

Formally, an $n$-state AfA $ M $ is a 5-tuple
\[
	M = (E, \Sigma, \{ A_\sigma \mid \sigma \in \Sigma \cup \{ \dollar  \} \},v_0,E_a \}),
\] 
where $ E = \{e_1,\ldots,e_n\} $ is the set of states, $ E_\sigma $ is an affine transformation that is applied on the actual affine state when reading symbol $ \sigma \in \Sigma \cup \{ \dollar \} $, $ v_0 $ is the initial affine state, and $ E_a \subseteq E $ is the set of accepting states. The AfA $ M $ starts its computation with $ v_0 $. Let $ x \in \Sigma^* $ be a given input. The computation of $ M $ can be traced similar to MCQFAs:
\[
	v_j = A_{x_j} v_{j-1},
\]
where $ 1 \leq j \leq |x| $. Then, the right end-marker is read for post-processing. (Remark that the right end-marker is unnecessary for MCQFAs \cite{BP02}):
\[
	v^x_f = v_f = A_\dollar v_{|x|}. 
\]
After this, a weighting operator is applied, which gives the probability of observing a state as the normalized weight of its value.  Here the weight of each state is the absolute value of the corresponding entry and the weight of $ v_f $ is the $ l_1 $-norm of $ v_f $, denoted $ |v_f| $. So, we observe the $ j $-th state with probability
\[
	\dfrac{| v_f[j] |}{ | v_f | }.
\]
Therefore, the input $  x $ is accepted by $M$ with the following probability
\[
	\dfrac{\sum\limits_{e_j \in E_a} \left| v_f[j] \right| }{ \left| v_f \right| }
	=
	\dfrac{\sum\limits_{e_j \in E_a} \left| v_f[j] \right| }{ \left| v_f[1] \right| + \cdots + \left| v_f[n] \right| }.
\]

\subsection{Promise problems}

The disjoint languages $ X \subseteq \Sigma^* $ and $ Y \subseteq \Sigma^* $  are said to be separated by  $M$ exactly or zero-error if any $ x \in X $ is accepted by $M$ with probability 1 and any $y \in Y $ is accepted by $M$ with probability 0, or vice versa. If $ |X| = |Y| =1 $, then it is said that given two different words (or pair) are separated by $ M $ exactly. In case of one-sided bounded error, any $ x \in X $ is accepted with probability 1 and any $ y \in Y $ is accepted with probability at most $ p<1 $, or vice versa. If $ |X| = |Y| =1 $, then it is said that given two different words (or pair) are separated by $ M $ with one-sided bounded-error.

Nondeterministic QFA is a theoretical model and it is defined as a special acceptance mode of a QFA, also known as recognition with cutpoint 0 \cite{YS10A}. The disjoint languages $ X \subseteq \Sigma^* $ and $ Y \subseteq \Sigma^* $  are said to be separated by a nondeterministic MCQFA $M$ if any $ x \in X $ is accepted by $M$ with some nonzero probability and any $y \in Y $ is accepted by $M$ with probability 0, or vice versa. If $ |X| = |Y| =1 $, then it is said that given two different words are separated by nondeterministic $ M $.

\section{Separating pairs with 2 states}
\label{sec:pairs}

In this section, we present our results on separating pairs.

\subsection{MCQFAs with real amplitudes}
\label{sec:pair-MCQFA-real}

We start with a 2-state ($ \{q_1,q_2\} $) unary MCQFA defined on $ \mathbb{R}^2 $. Note that any possible quantum state of such automaton is a point on the unit circle, where $ \ket{q_1} $ is $(1,0) $ and $ \ket{q_2} $ is $(0,1)$ \cite{ShurY14A,ShurY15A}. For any given two integers $d \geq 0$ and $t>0$, $ R_{d,t} $ is such a MCQFA with the following the specifications, where $ R $ stands for rotation:
\begin{itemize}
	\item The initial state is $ \cos (\frac{d\pi}{2t}) \ket{q_1} - \sin(\frac{d\pi}{2t}) \ket{q_2}  $, the point on the unit circle obtained by making a clockwise rotation with angle $ \frac{d\pi}{2t} $ ($d$ times $  \frac{\pi}{2t} $) when starting at the point $ \ket{q_1} $.
	\item The single unitary operator is a counter-clockwise rotation with angle $ \frac{\pi}{2t} $.
	\item The single accepting state is $q_1$.
\end{itemize}
We represent the details of $ R_{d,t} $ in Figure \ref{fig:MCQFA}.

\begin{figure}[!ht]  
  \centerline {\includegraphics[width=0.5\textwidth]{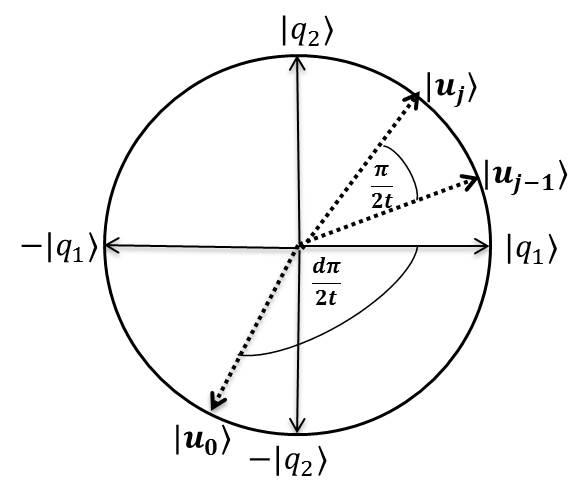}}
  \caption{The details of  $ R_{d,t} $}
\vspace*{-4mm}
\label{fig:MCQFA}
\end{figure}

\begin{theorem}
	Any given pair of unary words $ (a^d,a^{d+t}) $ ($ d \geq 0 $ and   $ t>0 $) can be exactly separated by MCQFA $ R_{d,t} $.
\end{theorem}
\begin{proof}
	As can be easily verified on Figure \ref{fig:MCQFA}, after reading $ a^d $, the automaton is in $ \ket{q_1} $ and so it is accepted with probability 1, and, after reading $ a^{d+t} $, the automaton is in $ \ket{q_2} $ and so it is accepted with probability 0.
\qed\end{proof}

\begin{corollary}
	Any easy pair of words can be separated exactly by a 2-state MCQFA with real amplitudes.
\end{corollary}

There exist hard pairs of words that can be exactly separated by a 2-state MCQFA with real amplitudes, for instance, the pair $(ab,ba)$: Let $ \left( \frac{1}{\sqrt{2}} ~~ \frac{1}{\sqrt{2}} \right)^T $ be the initial state, and we apply $ U_a $ and $ U_b $ when reading symbols $a$ and $b$, respectively, where
\[
	U_a = \mymatrix{rrr}{\frac{1}{\sqrt{2}} &~~& \frac{1}{\sqrt{2}} \\ \frac{1}{\sqrt{2}} & ~&  -\frac{1}{\sqrt{2}}}
	\mbox{ and }
	U_b = \mymatrix{rrr}{1 & ~ & 0 \\ 0 & & -1}.
\]
Then, after reading the words $ ab $ and $ ba $, we obtain the following final states:
\[
	\ket{u_f^{ab}} = 
		\mymatrix{rrr}{\frac{1}{\sqrt{2}} &~~& \frac{1}{\sqrt{2}} \\ \frac{1}{\sqrt{2}} & ~&  -\frac{1}{\sqrt{2}}} 
		\mymatrix{rrr}{1 & ~ & 0 \\ 0 & & -1}
		\myvector{\frac{1}{\sqrt{2}} \\ \frac{1}{\sqrt{2}}}
		=
		\myvector{0 \\ 1}
\]
and
\[
	\ket{u_f^{ba}} = 
		\mymatrix{rrr}{1 & ~ & 0 \\ 0 & & -1}
		\mymatrix{rrr}{\frac{1}{\sqrt{2}} &~~& \frac{1}{\sqrt{2}} \\ \frac{1}{\sqrt{2}} & ~&  -\frac{1}{\sqrt{2}}} 		
		\myvector{\frac{1}{\sqrt{2}} \\ \frac{1}{\sqrt{2}}}
		=
		\myvector{1 \\ 0}.
\]
Therefore, the pair $(ab,ba)$ can be exactly separated by 2-state MCQFAs with real amplitudes.

However, such automata cannot distinguish all pairs of words, as exemplified by the following simple result.

\begin{theorem}
	No 2-state nondeterministic MCQFA with real entries can separate two words $x,y\in \{a^2, b^2\}^*$ provided that $|x|_a = |y|_a$ and $|x|_b = |y|_b$.
\end{theorem}
\begin{proof}
Consider any such MCQFA, and let $U_a$ and $U_b$ be transition matrices corresponding to $a$ and $b$, respectively.
The operators $U_a^2$ and $U_b^2$ are rotations in $\mathbb{R}^2$, hence, they commute.
Thus,
\[
\ket {u^x_f} 
= U_b^{|x|_b} U_a^{|x|_a} \ket {u_0} 
= U_b^{|y|_b} U_a^{|y|_a} \ket {u_0} 
= \ket {u^y_f},
\]
and no final measurement can distinguish these two identical final states.
\qed
\end{proof}

\subsection{MCQFAs with complex amplitudes}
\label{sec:pair-MCQFA-complex}

In the previous section, we show that 2-state MCQFAs with real entries cannot separate all pairs of words. Here we conjecture that 2-state MCQFAs with complex entries actually can exactly separate any pair of words. In this section, we briefly sketch the relation of this conjecture to other known results.

\newcommand{\free}{\mathbf{F}}
The problem of separating pairs of words with MCQFAs is closely related to the problem of surjectivity of word maps in the special unitary group $SU(n)$.
Recall that $SU(n)$ is the group of unitary $n\times n$ matrices with determinant 1.

Let $\free_2$ be the free group on two letters $a$ and $b$.  
The group $\free_2$ consists of finite words over $\{a,b,a^{-1}, b^{-1}\}$ with the concatenation operation, modulo the relations $aa^{-1} = a^{-1}a = bb^{-1} = b^{-1}b = \epsilon$, where $\epsilon$ is the empty word.

Let $G$ be a group, usually a Lie group.
Any word $\omega \in \free_2$ defines in a natural way the corresponding word map $f_\omega\colon G\times G \to G$, in which a pair $(x, y)\in G^2$ is mapped to the product of the elements in G obtained from $\omega$ by replacing each $a$ by $x$ and each $b$ by $y$.  
In our case, we choose $G = SU(2)$, and we are interested in the image of the word map.
Recall that $SO(3) = SU(2)/\{\pm I\}$, so one might also consider the special orthogonal group in 3 dimensions instead of $SU(2)$.

\begin{fact}
\label{fact:1}
A pair of words $x,y\in \{a,b\}^*$ can be separated by a 2-state MCQFA if and only if the image of the word map corresponding to $\omega = xy^{-1} \in \free_2$ in $SU(2)$ contains a rotation by $\dfrac\pi2$, i.e., the unitary $\begin{pmatrix}i&0\\0&-i\end{pmatrix}$.
\end{fact}

\begin{proof}
For all unitaries $U_a$ and $U_b$ and $V$, we have 
\[
f_\omega(V^\dagger U_aV, V^\dagger U_bV) = V^\dagger f_\omega(U_a, U_b)V.
\]
Next, $f_\omega$ is a continuous map, and $SU(2)$ is compact and connected.  Thus, the image $f_\omega$ is of the form
\[
\left\{   
U\in SU(2) \mid \text{$U$ has eigenvalues $e^{\pm i\theta}$ with $0\le \theta\le \alpha$}
\right\}
\]
for some real $\alpha = \alpha(\omega)$ dependent on $\omega$.
The matrix $\begin{pmatrix}i&~~~0\\0&-i\end{pmatrix}$ is in the image of $f_\omega$ if and only if $\alpha(\omega) \ge \pi/2$.

Let $x = x_1\cdots x_n$ and $y = y_1\cdots y_m$.
For a MCQFA with the initial state $\ket{u_0}$ and transition matrices $U_a$ and $U_b$, we have
\[
\braket{u^x_f}{u^y_f}
 = \bra{u_0} U^\dagger_{x_1}U^\dagger_{x_2}\cdots U^\dagger_{x_n} U_{y_m}U_{y_{m-1}}\cdots U_{y_1} \ket{u_0} 
= \bra{u_0} f_\omega(U_a^\dagger, U_b^\dagger) \ket{u_0}. 
\]

If $\alpha(\omega)\ge\pi/2$, take $U_a$ and $U_b$ such that $f_\omega(U_a^\dagger, U_b^\dagger) = \begin{pmatrix}i&0\\0&-i\end{pmatrix}$ and $\ket{u_0} = \begin{pmatrix}1\\1\end{pmatrix}$.
In this case, $\braket{u^x_f}{u^y_f} = 0$, and these states can be separated exactly.

If $\alpha(\omega)<\pi/2$, then it is easy to see that $\bra{u_0} f_\omega(U_x^\dagger, U_y^\dagger) \ket{u_0} >0$ for all $U_x$, $U_y$ and $u_0$.  Hence, the states $u_f^x$ and $u_f^y$ cannot be separated exactly.
\qed
\end{proof}

Properties of word maps have been studied in various settings.
For instance, it is well known that there exist $U_a, U_b \in SO(3)$ such that the map $\omega \mapsto f_\omega(U_a,U_b)$ is an injective group homomorphism from $\free_2$ to $SO(3)$---the construction underlying the Banach-Tarski paradox.
For instance, one may take the matrices
\[
	U_a = \dfrac{1}{5} \left( \begin{array}{rrr}
		4 & 3 & 0 \\
		-3 & 4 & 0 \\
		0 & 0 & 5
\end{array}	 \right)
	\quad \mbox{ and }\quad
	U_b = \dfrac{1}{5} \left( \begin{array}{rrr}
		4 & 0 & 3 \\
		0 & 5 & 0 \\
		-3 & 0 & 4
\end{array}	 \right).
\]
In particular, for each $\omega\ne\epsilon$ the image of $f_\omega$ in $SU(2)$ is non-trivial.
Similar to Fact~\ref{fact:1}, we can get the following result.

\begin{theorem}
	\label{thm:nonMCQFA}
	Any pair $(x,y)$ can be separated by a 2-state nondeterministic MCQFA.
\end{theorem}

A famous result by Borel~\cite{borel:freeSubgroups} states that the image of $f_\omega$ is dense in the Zariski topology whenever $G$ is an algebraic connected semi-simple group and $\omega \ne \epsilon$.
However, for $G = SU(n)$, this does not imply that the image is dense in the ordinary topology.
Actually, it can be very far from that.
As shown by Thom~\cite{thom:convergentSeries}, for every $n\in\mathbb{N}$ and any neighbourhood $\cal V$ of $I\in SU(n)$ in the ordinary topology, there exists $\omega\in \free_2$ such that the image of $f_\omega$ is contained in $\cal V$.

However, the last result does not rule out the approach based on Fact~\ref{fact:1}, since the word $xy^{-1}$ has special structure.  In particular, as we show shortly, it is located shallow in the derived series of $\free_2$.

If $G$ is a group and $x,y\in G$, then $[x,y] = xyx^{-1}y^{-1}$ is called the commutator of $x$ and $y$.
The derived subgroup $G^{(1)}$ of $G$ is the subgroup $[G,G]$ generated by all the commutators of $G$.
By induction, one can define the $n$-th derived subgroup $G^{(n)} = [G^{(n-1)}, G^{(n-1)}]$.
This should not be confused with the lower central series, defined by $G_1 = G$ and $G_n = [G_{n-1}, G]$.

\begin{fact}
For any two different words $x,y\in\{a,b\}^*$, the element $xy^{-1}\in \free_2$ lies outside of the second derived group $\free_2^{(2)}$.
\end{fact}

\begin{proof}
An element $\omega\in \free_2$ lies in the first derived group $\free_2^{(1)}$ iff the total degree of both $a$ and $b$ in $\omega$ is 0.
That is, $xy^{-1}\notin \free_2^{(1)}$ if and only if $(x,y)$ is an easy pair.

Now assume that $(x,y)$ is a hard pair.
It is well-known that $\free_2^{(1)}$ is a free group with the set of generators 
\[
T = \left\{
[a^k, b^\ell] \; \mid\; k,\ell\in\mathbb Z\setminus\{0\}
\right\}.
\]
Note that $[a^k, b^\ell]^{-1} = [b^\ell, a^k]$.
Again, $\omega \in \free_2^{(1)}$ lies in $\free_2^{(2)}$ iff the unique decomposition of $\omega$ into the elements of $T$ contains each $[a^k, b^\ell]$ with total degree 0.

Note that for $x\in \{a,b\}^*$, we have a decomposition of the form
\[
x = \prod_i [a^{k_i}, b^{\ell_i}]^{\varepsilon_i} \cdot a^{|x|_a} b^{|x|_b},
\]
where, for all $i$, we have $k_i + \ell_i > k_{i-1} + \ell_{i-1}$ and $\varepsilon_i =\pm 1$.
Now it is not hard to see that $xy^{-1} \in \free_2^{(2)}$ if and only if $x = y$.
\qed
\end{proof}

There are some results showing surjectivity of $f_\omega$ for shallow $\omega$ in these series.
Using ideas by Got\^o~\cite{goto:theorem}, Elkasapy and Thom~\cite{elkasapy:goto} showed that if $\omega\notin \free_2^{(2)}$, then the corresponding word map $f_\omega\colon SU(n)\times SU(n) \to SU(n)$ is surjective for infinitely many $n\in\mathbb N$.
Another result by Klyachko and Thom~\cite{klyachko:equationsOverGroups} implies that the map $f_\omega\colon SU(2)\times SU(2) \to SU(2)$ is surjective if $\omega\notin [\free_2, \free_2]^2[\free_2, [\free_2, \free_2]]$.
The latter result does not always apply in our case since, for example,
\[
abba(baab)^{-1} = [ab, [b, a]] \in [\free_2, [\free_2, \free_2]] \subseteq [\free_2, \free_2]^2[\free_2, [\free_2, \free_2]].
\]

However, as noted in~\cite{elkasapy:goto}, no word $\omega\notin \free_2^{(2)}$ and $n\in\mathbb{N}$ is known such that $f_\omega$ is not surjective in $SU(n)$.
If $f_\omega\colon SU(2)\times SU(2)\to SU(2)$ is surjective for every $\omega\notin \free_2^{(2)}$, then any two different words can be exactly separated by a 2-qubit MCQFA.
This provides additional motivation to study this type of word maps.

\subsection{AfAs}
\label{sec:pair-AfA}

Here we show that 2-state AfAs can separate any pair exactly. We start with a 2-state AfA for easy pairs since we find the method algorithmically interesting.

Let $ S_{d,t} $ be a 2-state unary AFA that does not use right end-marker (or the related operator is identity), where $ S $ stands for subtraction. The only accepting state is the first state. The initial affine state is 
\[
	v_0 = \myvector{ 1 + \dfrac{d}{t} \\ \\ -\dfrac{d}{t} }
\]
and the affine operator for symbol $ a $ is
\[
	A_a = \mymatrix{ccc}{ 1 - \dfrac{1}{d} & ~~ &  -\dfrac{1}{d} \\ \\ \dfrac{1}{d} & & 1 + \dfrac{1}{d}  }.
\]
The affect of $ A_a $ can be easily observed after applying to the initial state:
\[
	v_1 = A_a v_0 = \mymatrix{ccc}{ 1 - \dfrac{1}{d} & ~~ &  -\dfrac{1}{d} \\ \\ \dfrac{1}{d} & & 1 + \dfrac{1}{d}  } \myvector{ 1 + \dfrac{d}{t} \\ \\ -\dfrac{d}{t} } = \myvector{ 1 + \dfrac{d}{t} - \dfrac{1}{d} \\ \\ -\dfrac{d}{t} + \dfrac{1}{d} }.
\]
And, it can be iteratively shown that
\[
	v_j = A^j_a v_{j-1} = \myvector{ 1 + \dfrac{d}{t} - \dfrac{j}{t} \\ \\ -\dfrac{d}{t} + \dfrac{j}{t} },
\]
where $ j > 0 $. 

\begin{theorem}
	Any given two different unary words $ a^d $ and $ a^{d+t} $ ($ d \geq 0 $ and   $ t>0 $) can be separated by AfA $ S_{d,t} $ exactly.
\end{theorem}
\begin{proof}
	The AfA $ S_{d,t} $ starts in $ v_0 $, and after reading $ d $ symbols, the value of the first state is subtracted by $ \frac{d}{t} $, and so, the affine state becomes
	\[
		v_d = \myvector{1 \\ 0}
	\]
	and the input is accepted exactly. If $ t $ more symbols are read, the value of first state is subtracted by $ \frac{t}{t} $, and so, the affine state becomes
	\[
		v_{d+t} = \myvector{0 \\ 1}
	\]
	and the input is accepted with zero probability.
\qed\end{proof}

\begin{corollary}
	Any easy pair can be separated by a 2-state AfAs exactly.
\end{corollary}

Let $ E_{x,y} $ be a 2-state AfA defined on $ \Sigma = \{0,1\} $ for the binary pair $ (x,y) $, where $ E $ stands for encoding. The only accepting state is the second one. We denote the value of binary number $ 1x $ as $ e(x) $. We set $ d = e(x)-e(y) $. Let $ z \in \{0,1\}^* $ be the given input. The aim of $ E_{x,y} $ is to set the value of the first state to $ e(z) $ after reading $ z $. The initial state is 
\[
	v_0 = \myvector{1 \\ 0},
\]
where the first digit of $ e(z) $ is already encoded. Then, we apply the following operators for each $ 0 $ and $ 1 $:
\[
	A_0  = \mymatrix{rr}{ 2 & 0 \\ -1 & 1 } \mbox{ and } A_1 = \mymatrix{rr}{ 3 & 1 \\ -2 & 0 }.
\]
Suppose that a prefix of $z$, say $ z' $, is encoded, after reading $z'$, and the affine state is as follows:
\[
	v_{|z'|} = \myvector{ e(z') \\ 1 - e(z') }.
\]
If the next symbol is $ 0 $, then the new affine state will be
\[
	v_{|z'|+1} = \mymatrix{rr}{ 2 & 0 \\ -1 & 1 }  \myvector{ e(z') \\ 1 - e(z') } = \myvector{ 2e(z') \\ 1 - 2e(z') } = \myvector{ e(z'0) \\ 1 - e(z'0) },
\]
and if the it is $ 1 $, then 
the new affine state will be
\[
	v_{|z'|+1} = \mymatrix{rr}{ 3 & 1 \\ -2 & 0 }  \myvector{ e(z') \\ 1 - e(z') } = \myvector{ 2e(z') +1 \\ - 2e(z') } = \myvector{ e(z'1) \\ 1 - e(z'1) }.
\]
So, the affine operators work as desired. Therefore, after reading $ z $, the affine state is
\[
	v_{|z|} = \myvector{ e(z) \\ 1 - e(z) }.
\]
On the right end-marker, we apply a composition of two affine operators:
\[
	A_\dollar = A''_\dollar A'_\dollar =  \mymatrix{ ccc }{ \dfrac{1}{d} & ~~ & 0 \\ \\ 1 - \dfrac{1}{d} & & 1 } \mymatrix{ ccc }{ e(x) -1 & ~~ & e(x) \\ 2-e(x) & & 1 - e(x) }.
\]
After applying $ A'_\dollar $ and $ A''_\dollar $, we obtain respectively
\begin{equation}
 	\label{eq:affine-dollar1}
 	\myvector{ e(x) - e(z) \\ 1 - e(x) + e(z) } = \mymatrix{ ccc }{ e(x) -1 & ~~ & e(x) \\ 2-e(x) & & 1 - e(x) } \myvector{ e(z) \\ 1 - e(z) }
\end{equation}
and
\[
	v_f = \mymatrix{ ccc }{ \dfrac{1}{d} & ~~ & 0 \\ \\ 1 - \frac{1}{d} & & 1 } \myvector{ e(x) - e(z) \\ 1 - e(x) + e(z) } = 
	\myvector{ \dfrac{e(x)-e(z)}{d} \\ \\ 1 - \dfrac{e(x)-e(z)}{d}  }.
\]

Now, we can analyse the behaviour of $ E_{x,y} $ on $ (x,y) $. If $ z = x $, then $ v_f = \myvector{ 0 \\ 1 } $ and so the input is accepted with probability 1. If $ z=y $, then $ v_f = \myvector{ \frac{d}{d} = 1 \\ 0 } $ and so the input is accepted with probability 0.

\begin{theorem}
	Any pair is separated by 2-state AfAs exactly.
\end{theorem}

In the remaining part, we present a one-sided bounded error algorithm for AfAs that can recognize a singleton language: $ L_x = \{ x \in \Sigma^* \} $ or separating $x$ from any other word. Remark that if the alphabet is unary ($ |\Sigma| = 1 $), then it is called counting problem \cite{KTSV99} and it can be solved by 2-state AfAs with one-sided bounded-error \cite{VilY16A}. Now, we present an algorithm for $ L_x $ based on the AfA $ E_{x,y} $ given in the previous section.

Let $ B_x $ defined on $\Sigma = \{0,1\} $ be a 2-state AfA, where $ B $ stands for bounded-error. The AfA $ B_x $ is identical to $ E_{x,y} $ except that the affine operator for the right end-marker is only $ A'_\dollar $. Therefore, if the input $ z $ is read, the final affine state (see Equation \ref{eq:affine-dollar1}) is
\begin{equation}
	\label{eq:affine-final-B}
	v^z_f = \myvector{ e(x) - e(z) \\ 1 - e(x) + e(z) }.
\end{equation}
Then, if $ z=x $, the input is accepted with probability 1. Otherwise, $ e(z) - e(x) = i $ is a non-zero integer and so the final state will be one of the followings:
\[
	v_f = \myvector{i \\ 1 - i} \in \left\{ \cdots, \myvector{ -2 \\ 3 }, \myvector{ -1 \\ 2 }, \myvector{ 1 \\  0 }, \myvector{ 2 \\ -1 }, \cdots \right\}
\]
and the accepting probability is at most $ \frac{2}{3} $.

\begin{theorem}
	For any given word $ x \in \Sigma^* $, 2-state AfAs can separate $x$ from any other word with one-sided bounded error.
\end{theorem}
\begin{proof}
	The proof for unary case is given in \cite{VilY16A} and the binary case is given above. For larger alphabets, we use the above algorithm. If $ |\Sigma|>k $, $ e(z) $ denote the value of $k$-ary number $ 1z $ for the word $ z \in \Sigma^* $ and the AfA can encode $ 1z $ in the value of the first state. The affine transformation for symbol $ \sigma \in \{ 0,1,\ldots,k-1 \} $ is
	\[
		A_\sigma = \mymatrix{ccc}{k+\sigma & ~~ & \sigma \\ 1 - k -\sigma & & 1 - \sigma}.
	\]
	If we apply $ A_\sigma $ to $ \myvector{ e(z) \\ 1- e(z) } $, we obtain
	\[
		\myvector{ e(z\sigma) = ke(z) + \sigma \\ 1 - e(z\sigma) } = \mymatrix{ccc}{k+\sigma & ~~ & \sigma \\ 1 - k -\sigma & & 1 - \sigma} \myvector{ e(z) \\ 1- e(z) }.
	\]
\qed\end{proof}

One may ask how can we reduce the undesired accepting probability for the words other than $ x $. In Equation \ref{eq:affine-final-B}, we can easily modify $ A'_\dollar $ and  obtain the following state for some $k>1 $:
\[
	v_f^z = \myvector{ k(e(x)- e(z)) \\ 1 - k(e(x)-e(z)) }.
\]
The input $ x $ is still accepted with probability 1. The final state for any other input $ z $ ($ e(x) - e(z) = i  $) can be one of the followings
\[
	\myvector{ki \\ 1 - ki} \in \left\{ \cdots, \myvector{ -2k \\ 1+2k }, \myvector{ -k \\ 1+k }, \myvector{ 0 \\  1 },\myvector{ k \\  1-k }, \myvector{ 2k \\ 1-2k }, \cdots \right\}
\]
and so the undesired accepting probability is at most $ \frac{k+1}{2k+1} = \frac{1}{2} + \frac{1}{4k+2} $, which can be arbitrary close to $ \frac{1}{2} $ when $ k \rightarrow \infty $.

Now, we show that the undesired accepting probability can be decreased arbitrarily by 3-state AfAs. We call it \textit{3-state trick}. We modify $ B_x $ by using another state such that the final state will be 
\[
	v_f = v^z_f = \myvector{ -ki \\ (k+1)i \\ 1-i }.
\]
Then, if $ x=z $, $ i=0 $ and so the input is accepted with probability 1. If $ x \neq z $, $ i $ is a nonzero integer, and the final state will be 
\[
	v_f \in \left\{ \cdots, \myvector{ 2k \\ -2k-2 \\ 3  }, \myvector{ k \\ -k-1 \\ 2  },  \myvector{ -k \\ k+1 \\ 0  }, \myvector{ -2k \\ 2k+2 \\ -1  },\cdots  \right\}
\]
and the accepting probability is at most $ \frac{2}{2k+1} $, which can be arbitrary close to 0 when $ k \rightarrow \infty $.

\begin{theorem}
	For any given word $ x \in \Sigma^* $, 3-state AfAs can separate $x$ from any other word with arbitrary small one-sided bounded error.
\end{theorem}

\section{Separating two finite sets}
\label{sec:finite-sets}

In this section, we focus on a more general problem: Separating two finite languages. Let $ X = \{ x_1,\ldots,x_m \} $ and $ Y = \{ y_1,\ldots,y_n \} $ be two disjoint set of binary words by assuming that $ m \leq n $ (the sets are exchanged, otherwise).

\subsection{Exact AfAs}

If $ 1 = m < n $, then we can still obtain an exact algorithm. For this purpose, we use the AfA $ E_{x,y} $ given in Section \ref{sec:pair-AfA} after some certain modifications. Remember that the final state of $ E_{x,y} $ on the input $ z $ will be
\[
	v_f^z = \myvector{ e(x)-e(z) \\ 1  - e(x) + e(z)  }
\]
if $ A''_\dollar $ is not applied. Since we do not use the division factor (that uses $y$), we call this modified version $ E'_{x} $. 

We design a new AfA $ E_{x,Y} $ that executes $ E'_{y_1}, E'_{y_2},\ldots,E'_{y_n} $ in parallel: all initial states and affine operators for the same symbols are tensored in the same order. Then new final state on the input  $ z $ will be
\[
	\myvector{ e(y_1)-e(z) \\ 1 - e(y_1) + e(z)  } \otimes \myvector{ e(y_1)-e(z) \\ 1 - e(y_1) + e(z)  } \otimes \cdots \otimes \myvector{ e(y_n)-e(z) \\ 1  - e(y_n) + e(z)  } 
\]
which has a single vector representation having the first entry as 
\[
	\Pi^n_{j=1} (  e(y_j) - e(z) ) .
\]
So, if we sum up all entries except the first one to the second entry, the new final state will be 
\[
	\myvector{ \Pi_{j=1}^n (  e(y_j) - e(z) ) \\ 1- \Pi_{j=1}^n (  e(y_j) - e(z) ) \\ 0 \\ \vdots \\ 0 }
\]
Now, we define $ D $ as $ \Pi_{j=1}^n ( e(y_j) - e(x)) \neq 0 $. We can modify the final state once more as follows:
\[
	\myvector{  \dfrac{ \Pi_{j=1}^n (  e(y_j) - e(z) ) } {D}  \\ \\ 1- \dfrac{ \Pi_{j=1}^n (  e(y_j) - e(z) ) } {D} \\ \\ 0 \\ \vdots \\ 0 }
\]
Lastly, we pick the first state as the only accepting state.

If $ z = x $, then the first entry of the final state will be 1 and all the other entries are zeros. So, the input $ x $ is accepted with probability 1. If $ z \in Y $, then the first entry will be zero and so any member of $ Y $ is accepted with probability $ 0 $.

\begin{theorem}
	A given word $ x $ is separated from any member of a finite language $ Y $ by a $ 2^{|Y|} $-state AfA exactly.
\end{theorem}

\subsection{Bounded-error AfAs}

In this section, we focus on the general case $ 2 \leq m \leq n $. We use the bounded error algorithm for singleton language given at the end of Section \ref{sec:pair-AfA}. Remember that $ B_x $ is the AfA that leaves the final state in
\[
	\myvector{ e(x) - e(z) \\ 1 - e(x) + e(z) \\  }
\]
after reading input $ z $ (Equation \ref{eq:affine-final-B}). Let $ B_X $ be the $ 2^m $-state AfA obtained by tensoring $ B_x $s for each $ x \in X $ and the first state be the only accepting state. Then, its final state for the input $z$ will be
\[
	v_f = \myvector{ \Pi_{j=1}^m (e(x_j)-e(z))  \\ \vdots \\ \vdots},
\]
It is clear that the value of the first state will be zero if $ z \in X $, and it will be nonzero, otherwise.

We can easily modify $ B_X $, say $ B'_X $, in order to obtain the following final state
\[
	v_f = \myvector{ 1 - \Pi_{j=1}^m (e(x_j)-e(z)) \\ \Pi_{l=1}^m (e(x_j)-e(z)) \\ 0 \\ \vdots \\ 0 \\ }.
\]
Thus, the value of the first state will be 1 if $ z \in X $, and it will be different than 1 if $ z \notin X $. Therefore, the inputs in $ X $ are accepted with probability 1 and all the other inputs are accepted with probability at most $ \frac{2}{3} $. By using 3 state tricks given at the end of Section \ref{sec:pair-AfA}, the undesired accepting probability can be arbitrary close to zero (without adding a new state).

\begin{theorem}
	The disjoint languages $ X $ and $ Y $ ($ 1 <|X| \leq |Y| $) can be separated by $ 2^{|X|} $-state AfAs with arbitrary small one-sided bounded-error.  
\end{theorem}

\subsection{Nondeterministic MCQFAs}

Let $ \Sigma = \{a,b\} $ and $ x \in \Sigma^* $ be a word. We define a 3-state $ (\{q_1,q_2,q_3\}) $ MCQFA with real amplitudes that uses the matrices given before Theorem \ref{thm:nonMCQFA} when reading an $a $ and $ b $:
\[
	U_a = \dfrac{1}{5} \left( \begin{array}{rrr}
		4 & 3 & 0 \\
		-3 & 4 & 0 \\
		0 & 0 & 5
\end{array}	 \right)
	\quad \mbox{ and }\quad
	U_b = \dfrac{1}{5} \left( \begin{array}{rrr}
		4 & 0 & 3 \\
		0 & 5 & 0 \\
		-3 & 0 & 4
\end{array}	 \right).
\]
The accepting states are $ \{q_2,q_3\} $. The initial state is
\[
	\ket{u_0} = U^{-1}_{x_{1}} U^{-1}_{x_{2}} \cdots U^{-1}_{x_{|x|}} \ket{q_1}.
\]
If $ M_x $ reads $x$, the final state will be
\[
	\ket{u^x_f} = U_{x_{1}} U_{x_{2}} \cdots U_{x_{|x|}} U^{-1}_{x_{1}} U^{-1}_{x_{2}} \cdots U^{-1}_{x_{|x|}} \ket{q_1} = \ket{q_1}.
\]
If $ M_x $ reads a different word, say $ y \in \Sigma^* $ ($ y \neq x $), the final state will be
\[
	\ket{u^x_f} = U_{y_{1}} U_{y_{2}} \cdots U_{y_{|y|}} U^{-1}_{x_{1}} U^{-1}_{x_{2}} \cdots U^{-1}_{x_{|x|}} \ket{q_1} =  \alpha_1 \ket{q_1} + \alpha_2 \ket{q_2} + \alpha_3 \ket{q_3},
\] 
where $ |\alpha_1|^2 < 1 $ and $ | \alpha_2 |^2 + | \alpha_3 |^2 > 0 $ \cite{AW02}. Thus, $ M_x $ accepts only $ x $ with probability 0 and it accepts any other word ($y \neq x$) with non-zero probability. As pointed in Section \ref{sec:pair-MCQFA-complex} (see also \cite{AW02}), $ M_x $ can be mapped to a 2-state $ ( \{p_1,p_2\} ) $ MCQFA with complex entries, say $ N_x $, such that (1) after $ N_x $ reads $ x $, the amplitude of $ p_2 $ in the final state is zero, and when it reads $ y \neq x $, the same amplitude is non-zero. 

\begin{corollary}
	The complement of any singleton language is recognized by 2-state nondeterministic MCQFAs.
\end{corollary}

Let $ N(X) = \{N_{x_1}, N_{x_2}, \ldots, N_{x_m} \} $ be the set of 2-state MCQFAs for words in $ X $. We can obtain a MCQFA, say $ N_X  $, by tensoring all  MCQFAs in $ N(X) $,
\[
	N_X = N_{x_1} \otimes N_{x_2} \otimes \cdots \otimes N_{x_m},
\] 
i.e. executing all of them in parallel. The set of states of  $ N_X $ is $ \{p_1,p_2\}^m $. If $ \ket{u_{j,0}} $ is the initial state of $ N_{x_j} $ and $ U_{j,a} $ ($U_{j,b}$) is the unitary operator for symbol $ a $ ($b$), then the initial state of $ N_X $ is
\[
	\ket{u_{1,0}} \otimes \ket{u_{2,0}} \otimes \cdots \otimes \ket{u_{m,0}} 
\]
and the unitary operator for symbol $a$ ($b$) is
\[
	U_{1,a} \otimes U_{2,a} \otimes \cdots \otimes U_{m,a} ~~~(U_{1,b} \otimes U_{2,b} \otimes \cdots \otimes U_{m,b}),
\]
where $ 1 \leq j \leq m $. Similarly, if $ \ket{u_{j,f}^y} $ is the final state of $ N_{x_j} $ and $ \beta_j $ is the amplitude of the state $ \ket{p_2} $ after reading binary word $ y  $, then the final state of $ N_X $ on $ y $ will be
\[
	\ket{u_{1,f}^y} \otimes \ket{u_{2,f}^y} \otimes \cdots \otimes \ket{u_{m,f}^y}
\]
and so the amplitude of $ \ket{(p_2,p_2,\ldots,p_2)} $ will be
\[
	\beta=\beta_1 \beta_2 \cdots \beta_m.
\]

It is clear that, if $ x_j = y $, then $ \beta $ will be zero since $ \beta_j $ is zero. More generally, $ \beta =0 $ if and only if $ y \in X $. Thus, by picking $ (p_2,p_2,\ldots,p_2) $ as the only accepting state of $  M_X $, we can obtain the machine that separates any given word from a word in $ X $. Remark that the number of states of $ N_X $ is $ 2^m $.

\begin{theorem}
	The disjoint binary languages  $ X $ and $ Y $ ($ 1 \leq |X| \leq |Y| $) can be separated by nondeterministic MCQFAs with $ 2^{|X|} $ states.
\end{theorem}

\section*{Acknowledgement.} We are grateful to Andreas Thom for pointing us to his papers on the subject, and for answering some of our questions.

\bibliographystyle{plain}
\bibliography{tcs}

\end{document}